\documentclass[a4paper,USenglish,numberwithinsect,cleveref,autoref,thm-restate,authorcolumns]{lipics-v2021}
\graphicspath{{./fig/}} %helpful if your graphic files are in another directory

\usepackage{thmtools}
\usepackage{thm-restate}

\usepackage{multirow}

\usepackage{mathtools}
\mathtoolsset{centercolon}
\usepackage{amsmath}
\usepackage{amssymb,amsfonts}
\usepackage{mathrsfs}
\usepackage{enumerate}

\usepackage[nocompress,noadjust]{cite} % this breaks the links when used with bibunits

\usepackage{hyperref,xcolor}
\usepackage{fvicommands}

\title{Extended MSO Model Checking via Small Vertex Integrity}
\author{Tatsuya Gima}
{Nagoya University, Nagoya, Japan %\and 
}
{gima@nagoya-u.jp}
{}
{}

\author{Yota Otachi}
{Nagoya University, Nagoya, Japan
\and \url{https://www.math.mi.i.nagoya-u.ac.jp/~otachi/}}
{otachi@nagoya-u.jp}
{https://orcid.org/0000-0002-0087-853X}
{JSPS KAKENHI Grant Numbers
  JP18H04091, % Uehara A
  JP18K11168, % Otachi C, old
  JP18K11169, % Kiyomi C
  JP20H05793, % Ito Henkaku B
  JP21K11752, % Otachi C, new
  JP22H00513. % Ono A
}

\authorrunning{T. Gima and Y. Otachi}

\Copyright{Tatsuya Gima and Yota Otachi}

\ccsdesc[500]{Mathematics of computing~Graph algorithms}
\ccsdesc[500]{Theory of computation~Parameterized complexity and exact algorithms}

\keywords{vertex integrity, monadic second-order logic, cardinality constraint, fixed-parameter tractability}

\category{}%optional, e.g. invited paper

\relatedversion{} %optional, e.g. full version hosted on arXiv, HAL, or other respository/website
\acknowledgements{The authors thank Michael Lampis and Valia Mitsou for fruitful discussions and sharing a preliminary version of \cite{LampisM21}.}%optional

\nolinenumbers %uncomment to disable line numbering

\hideLIPIcs  %uncomment to remove references to LIPIcs series (logo, DOI, ...), e.g. when preparing a pre-final version to be uploaded to arXiv or another public repository

\EventEditors{John Q. Open and Joan R. Access}
\EventNoEds{2}
\EventLongTitle{42nd Conference on Very Important Topics (CVIT 2016)}
\EventShortTitle{CVIT 2016}
\EventAcronym{CVIT}
\EventYear{2016}
\EventDate{December 24--27, 2016}
\EventLocation{Little Whinging, United Kingdom}
\EventLogo{}
\SeriesVolume{42}
\ArticleNo{23}
\begin{document}

% \begin{bibunit}

\maketitle

%%% Abstract, Keywords
% -*- coding:utf-8 -*-

\begin{abstract}
We study the model checking problem of 
an extended $\MSO$ with local and global cardinality constraints, called $\MSOGLLin$, introduced recently by 
Knop, Kouteck\'{y}, Masa\v{r}\'{\i}k, and Toufar~[\textit{Log.\ Methods Comput.\ Sci.}, 15(4), 2019].
We show that the problem is fixed-parameter tractable parameterized by vertex integrity,
where vertex integrity is a graph parameter standing between vertex cover number and treedepth.
Our result thus narrows the gap between 
the fixed-parameter tractability parameterized by vertex cover number and 
the W[1]-hardness parameterized by treedepth.
\end{abstract}

%%% Introduction
% -*- coding:utf-8 -*-

\section{Introduction}

\label{sec:intro}
One of the most successful goals in algorithm theory is
to have a meta-theorem that constructs an efficient algorithm
from a \emph{description} of a target problem in a certain format
(see e.g.,~\cite{HlinenyOSG08,GroheK09,Kreutzer11}).
Courcelle's theorem~\cite{Courcelle90mso1,Courcelle92mso3,BoriePT92,CourcelleMR00} is arguably the most successful example
of such an algorithmic meta-theorem, which says (with Bodlaender's algorithm~\cite{Bodlaender98}) that:
if a problem on graphs can be expressed in monadic second-order logic ($\MSO$),
then the problem can be solved in linear time on graphs of bounded treewidth.
Many natural problems that are NP-hard on general graphs
are shown to have expressions in $\MSO$ and thus have linear-time algorithms on graphs of bounded treewidth~\cite{ArnborgLS91}.

Although the expressive power of $\MSO$ captures many problems, 
it is known that $\MSO$ cannot represent some kinds of cardinality constraints~\cite{CourcelleE12}.
For example, it is easy to express the problem of finding a proper vertex coloring with $r$ colors
in $\MSO$ as the existence of a partition of the vertex set into $r$ independent sets,
where the length of the corresponding $\MSO$ formula depends on $r$.
However, the variant of the problem that additionally requires the $r$ independent sets to be of the same size cannot be expressed in $\MSO$ 
even if $r = 2$ (see \cite{CourcelleE12}).
Indeed, this problem is known to be W[1]-hard parameterized by $r$ and treewidth~\cite{FellowsFLRSST11}.\footnote{%
We assume that the readers are familiar with the concept of parameterized complexity.
For standard definitions, see e.g., \cite{CyganFKLMPPS15}.}
See \cite{Szeider10,BelmonteKLMO20,GimaHKKO21} for many other examples of such problems.

For those problems that do not admit $\MSO$ expressions and are hard on graphs of bounded treewidth,
there is a successful line of studies on smaller graph classes with more restricted structures.
For example, by techniques tailored for individual problems, several problems are shown to be tractable
on graphs of bounded vertex cover number~(see e.g., \cite{FellowsLMRS08,EncisoFGKRS09,FialaGK11}).
Such results are known also for more general parameters such as 
twin-cover~\cite{Ganian11}, neighborhood diversity~\cite{Lampis12}, and vertex integrity~\cite{GimaHKKO21}.
Then the natural challenge would be finding a meta-theorem covering (at least some of) such results.
Recently, such meta-theorems are intensively studied for extended $\MSO$ logics with ``cardinality constraints.''
In this paper, we follow this line of research and focus on vertex integrity as the structural parameter of input graphs.
The \emph{vertex integrity} of a graph is the smallest number $k=s+t$ such that
by removing $s$ vertices of the graph,
every component can be made to have at most $t$ vertices.
The concept of vertex integrity was introduced first in the context of network vulnerability~\cite{BarefootES87}.
It basically measures how difficult it is to break a graph into small components by removing a small number of vertices.
This can be seen as a generalization of vertex cover number,
which asks to remove vertices to make the graph edge-less (corresponding to the case $t=1$ of the definition of vertex integrity).
On the other hand, the concept of treedepth can be seen as a recursive generalization of vertex integrity.
Actually, their definitions give us the inequality
$\textrm{treedepth} \le \textrm{vertex integrity} \le \textrm{vertex cover number}-1$
for every graph~(see \cite{GimaHKKO21}).

There is another issue about Courcelle's theorem that the dependency of the running time on the parameters
(the treewidth of the input graph and the length of formula) is quite high~\cite{FrickG04}.
To cope with this issue, faster algorithms are proposed for special cases
such as vertex cover number, neighborhood diversity, and max-leaf number~\cite{Lampis12},
twin-cover \cite{Ganian11}, shrubdepth~\cite{GanianHNOMR12}, treedepth~\cite{GajarskyH15}, and vertex integrity~\cite{LampisM21}.
The methods in these results are similar in the sense
that they find a smaller part of the input graph that is equivalent to the original graph under the given $\MSO$ formula.
Interestingly, these techniques are used also in studies of extended $\MSO$ logics in these special cases.
Our study is no exception, and we use a result in~\cite{LampisM21} as a key lemma.

\paragraph*{Meta-theorems on extended $\MSO$ with cardinality constraints.}
In this direction, there are two different lines of research, which have been merged recently.
One line considers ``global'' cardinality constraints
and the other considers ``local'' cardinality constraints.

Recall that the property of having a partition into $r$ independent sets of equal size cannot be expressed in $\MSO$.
A remedy for this would be to allow a predicate like $|X| = |Y|$.
The concept of global cardinality constraints basically implements this but in a more general way
(see Section~\ref{sec:pre} for formal definitions).
It is known that the model checking for the extended $\MSO$ logic with global cardinality constraints is 
fixed-parameter tractable parameterized by neighborhood diversity~\cite{GanianO13}.

The concept of local cardinality constraints was originally introduced as the \emph{fairness} of a solution~\cite{LinS89a}.
The fairness of a solution (a vertex set or an edge set) upper-bounds 
the number of neighbors each vertex can have in the solution.
It is known that finding a vertex cover with an upper bound on the fairness is W[1]-hard
parameterized by treedepth and feedback vertex set number~\cite{KnopMT19}.
On the other hand, the problem of finding a vertex set satisfying an $\MSO$ formula and fairness constraints 
is fixed-parameter tractable parameterized by neighborhood diversity~\cite{MasarikT20} and by twin-cover~\cite{KnopMT19}.
The general concept of local cardinality constraint extends the concept of fairness
by having for each vertex, an individual set of the allowed numbers of neighbors in the solution.
It is known that the extension of $\MSO$ with local cardinality constraints admits
an XP  algorithm 
(i.e., a slicewise-polynomial time algorithm)
 parameterized by treewidth~\cite{Szeider11}.
%  (i.e., $n^{f(k)}$ time algorithm where $k$ is a treewidth and $f$ is a).

Knop, Kouteck\'{y}, Masa\v{r}\'{\i}k, and Toufar~\cite{KnopKMT19} recently converged two lines
and studied the model checking of extended $\MSO$
with both local and global cardinality constraints.
It is shown that the problem admits an XP algorithm parameterized by treewidth.
Furthermore, they showed that
the problem is fixed-parameter tractable parameterized by neighborhood diversity
if the cardinality constraints are ``linear,'' 
where each local cardinality constraint is a set of consecutive integers
and each global cardinality constraint is a linear inequality.

\paragraph*{Our results.}
We study the linear version of the problem in~\cite{KnopKMT19} mentioned above;
that is, the model checking of the extended $\MSO$ logic with linear local and global cardinality constraints.
We show that this problem, called \pname{$\MSOGLLin$ Model Checking}, is fixed-parameter tractable parameterized by vertex integrity.
This result fills a missing part in the map on the complexity of \pname{$\MSOGLLin$ Model Checking}
as vertex integrity fits between these parameters~\cite{GimaHKKO21,LampisM21} (see \cref{fig:paras}).
Note that by $\MSO$, we mean $\MSO_{1}$, which does not allow edge and edge-set variables.
After proving the main result, we show that the same result holds even for the same extension of $\MSO_{2}$.
We apply the results to several problems and show some new examples that are fixed-parameter tractable parameterized by vertex integrity.
We also show that some known results can be obtained as applications of our results.\footnote{Omitted from the conference version.
See \cref{sec:app}.}

\begin{figure}[htb]
\centering

\includegraphics[scale=1]{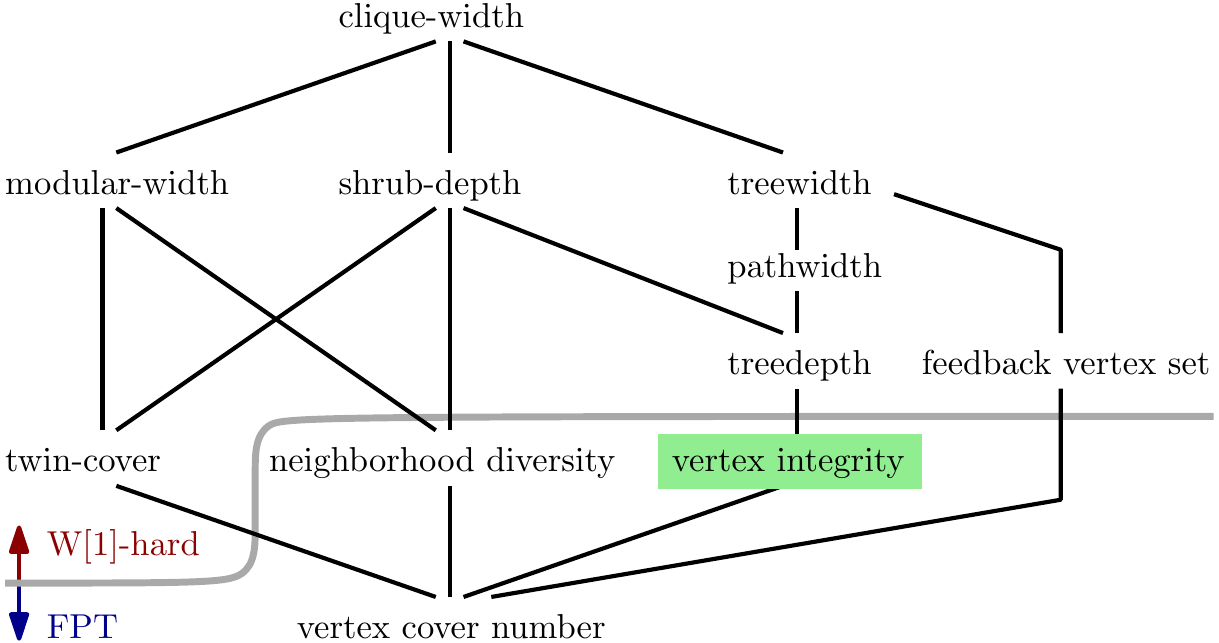}
\caption{Some of the major graph parameters and the complexity of \pname{$\MSOGLLin$ Model Checking}.
  If one parameter is an ancestor of another,
  then the ancestor is upper-bounded by a function of the descendant.
  The fixed-parameter tractability parameterized by neighborhood diversity is shown in~\cite{KnopKMT19}.
  The W[1]-hardness parameterized by twin-cover 
  and by treedepth and feedback vertex set number are shown in~\cite{KnopMT19}.}
\label{fig:paras}
\end{figure}

%%% Preliminaries
% -*- coding:utf-8 -*-

\section{Preliminaries}
\label{sec:pre}
For two integers $a$ and $b$, 
we define $[a,b] = \{x \in \mathbb{Z} \mid a \leq x \leq b\}$.
We write $[b]$ for the set $[1,b]$.
For two tuples $\vect{A} = (A_1, \ldots, A_p)$ and $\vect{B} = (B_1, \ldots, B_q)$,
the concatenation $(A_1, \ldots, A_p, B_1, \ldots, B_q)$ is denoted by
$\vect{A} \tupplus \vect{B}$.
For a function $f \colon X \to Y$ and a set $A \subseteq X$,
the restriction of $f$ to $A$ is denoted by $f|_A$.

\subsection{Graphs and colored graphs}
We consider undirected graphs without self-loops or multiple
edges.
Let $G=(V,E)$ be a graph.
The vertex set and the edge set of $G$
are denoted by $V(G)$ and $E(G)$, respectively.
A \emph{component} of $G$ is 
a maximal connected induced subgraph of $G$.
For a vertex set $S$ of a graph $G$,
the subgraph of $G$ induced by $V \setminus S$
is denoted by $G-S$.

A \emph{$p$-color list} $\cfun$ of $G$
is a tuple $\cfun= (C_1, \ldots, C_p)$ of $p$ 
vertex sets $C_i \subseteq V$.
Denote the set of colors assigned by $\cfun$ to $v\in V$ by $\cor_{\cfun}(v)$.
Note that each vertex can have multiple colors.
That is, $\cor_{\cfun}(v) = \{i \in [p] \mid v \in C_i\}$.
Note that $\cor_{\cfun}(v)$ can be computed in time polynomial in $|V|$ and $p$.
We call a tuple $(G,\cfun)$ a \emph{$p$-colored graph}.
If the context is clear, we simply call it a graph.

Let $\mathcal{G}_1 = (G_1, \cfun_1)$ and
$\mathcal{G}_2 = (G_2, \cfun_2)$ be $p$-colored graphs.
A bijection $\psi \colon V(G_1) \to V(G_2)$ is an \emph{isomorphism}
from $\mathcal{G}_1$ to $\mathcal{G}_2$
if $\psi$ satisfies the following conditions:
\begin{itemize}
  \item $\{u,v\} \in E(G_1)$ if and only if $\{\psi(u), \psi(v)\} \in E(G_2)$
  for all $u,v \in V(G_1)$;
  \item $\cor_{\cfun_1}(v) = \cor_{\cfun_2}(\psi(v))$ for all $v \in V(G_1)$.
\end{itemize}
We say that $\mathcal{G}_1$ and $\mathcal{G}_2$ are \emph{isomorphic} 
if such $\psi$ exists.

\subsection{Vertex integrity}

A \emph{$\vi(k)$-set} $S$ of a graph $G$ is a
set of vertices such that 
the number of vertices of every component of $G-S$
is at most $k-|S|$.
The \emph{vertex integrity} of a graph $G$,
denoted by $\vi(G)$,
is the minimum integer $k$ such that there is a 
$\vi(k)$-set of $G$.
In other words, it can be defined as follows:
\[
  \vi(G) = \min_{S \subseteq V(G)}\left\{|S| + \max_{C \in \mathrm{cc}(G-S)} |V(C)|\right\},
\]
where $\mathrm{cc}(G-S)$ is the set of connected components of $G-S$.
A $\vi(k)$-set of $G$, if any exists, can be found in $O(k^{k+1}n)$ time~\cite{DrangeDH16},
where $n$ is the number of vertices in $G$.

As mentioned above, the concept of vertex integrity was originally introduced in the context of network vulnerability~\cite{BarefootES87},
but recently it and its close relatives are used as structural parameters in algorithmic studies.
The \emph{safe number} was introduced with a similar motivation~\cite{FujitaMS16}
and later shown to be (non-trivially) equivalent to the vertex integrity in the sense
that the safe number is bounded if and only if so is the vertex integrity for every graph~\cite{FujitaF18}.
The definition of \emph{fracture number} is almost the same as the one for vertex integrity,
where the only difference is that it asks the maximum (instead of the sum) of the orders of $S$ and a maximum component of $G-S$ to be bounded by $k$.
The $\ell$-component order connectivity~\cite{DrangeDH16} measures the size of $S$ and the maximum order of a component of $G-S$ separately,
and defined to be the minimum size $k$ of a set $S$ such that each component of $G-S$ has order at most $\ell$.
For example, $1$-component order connectivity is exactly the vertex cover number.
Also, the $2$-component order connectivity is studied as the matching-splittability~\cite{JansenM15}.
A graph has vertex integrity at most $k$
if and only if the graph has $\ell$-component order connectivity at most $k - \ell$ for some $\ell$.

The fracture number was used to design efficient algorithms for \pname{Integer Linear Programming}~\cite{DvorakEGKO17},
\pname{Bounded-Degree Vertex Deletion}~\cite{GanianKO21},
and \pname{Locally Constrained Homomorphism}~\cite{BulteauDKOP22_arxiv}.
The vertex integrity was used in the context of subgraph isomorphism on minor-closed graph classes~\cite{BodlaenderHKKOO20},
and then used to design algorithm for several problems that are hard on graphs of bounded treedepth such as
\pname{Capacitated Dominating Set},
\pname{Capacitated Vertex Cover},
\pname{Equitable Coloring},
\pname{Equitable Connected Partition}\footnote{%
In~\cite{EncisoFGKRS09}, \pname{Equitable Connected Partition} was
shown to be W[1]-hard parameterized 
simultaneously by pathwidth, feedback vertex set number, and the number of parts.
In \cref{sec:w1}, we strengthen the W[1]-hardness by replacing pathwidth in the parameter with treedepth.},
\pname{Imbalance},
\pname{Maximum Common (Induced) Subgraph}, and
\pname{Precoloring Extension}~\cite{GimaHKKO21}.
A faster algorithm for \pname{$\MSO$ Model Checking} parameterized by vertex integrity is also known~\cite{LampisM21}.

\subsection{Monadic second-order logic}
A \emph{monadic second-order logic formula} (an \emph{$\MSO$ formula}, for short)
over $p$-colored graphs is a formula
that matches one of the following,
where $x$ and $y$ denote vertex variables, $X$ denotes a vertex-set variable,
$C_i$ denotes a vertex-set constant (color):
%-- Note: There is no vertex constant (label)
$E(x,y)$; $x = y$; $x \in X$; $x \in C_i$;
$\exists x.\varphi$, $\forall x.\varphi$, $\exists X.\varphi$, $\forall X.\varphi$,
$\varphi \land \psi$, $\varphi \lor \psi$, and $\lnot \varphi$,
where $\varphi$ and $\psi$ are $\MSO$ formulas.
These symbols have the following semantic meaning: $E(x,y)$ means that $x$ and $y$ are adjacent;
and the others are the usual ones.
Additionally, for convenience, we introduce $\MSO$ symbols $\true$ and $\false$ that 
are always interpreted as true and false, respectively.
Note that this version of $\MSO$ is often called $\MSO_1$. 
In Section~\ref{sec:ext}, we consider a variant called $\MSO_2$, % todo: short-ext -> ext (in the full version)
which has stronger expression power.

A variable is \emph{bound} if it is quantified and \emph{free} otherwise.
An $\MSO$ formula is \emph{closed} if
it has no free variables and \emph{open} otherwise.
We assume that every free variable is a set variable,
because a free vertex variable can be simulated by
a free vertex-set variable with an $\MSO$ formula expressing that the set is of size~1.
An \emph{assignment} of an open $\MSO$ formula $\varphi$ with $s$ free set variables
over $G$ is a tuple $\vect{X}^G =(X^G_1, \ldots, X^G_s)$
of $s$ vertex sets $X^G_i \subseteq V(G)$.
Let $\mathcal{G}$ be a $p$-colored graph, and $\varphi$ be an $\MSO$ formula.
If $\varphi$ is closed, we write $\mathcal{G} \models \varphi$ 
if $\mathcal{G}$ satisfies the property expressed by $\varphi$.
Otherwise, we write $(\mathcal{G}, \vect{X}^G) \models \varphi$ where $\vect{X}^G$ is an
assignment of $\varphi$
if $\mathcal{G}$ and $\vect{X}^G$ satisfies the property expressed by $\varphi$.

From the definition of $\MSO$,
one can see that
no $\MSO$ formula can distinguish 
isomorphic $p$-colored graphs.
See e.g.,~\cite{LampisM21} for a detailed proof.
\begin{lemma}[Folklore]
    Let $\mathcal{G}_1$ and $\mathcal{G}_2$ be isomorphic $p$-colored graphs.
    For every $\MSO$ formula $\varphi$, we have $\mathcal{G}_1 \models \varphi$ if and only if $\mathcal{G}_2 \models \varphi$.
    \label{lem:iso-equiv}
\end{lemma}

\subsection{Extensions of $\MSO$}
We introduce an extension of $\MSO$ proposed by
Knop, Kouteck\'{y}, Masa\v{r}\'{\i}k, and Toufar~\cite{KnopKMT19}.
Let $\varphi$ be an $\MSO$ formula with $s$ free set variables $X_1, \ldots, X_s$,
and $G$ be a graph with $n$ vertices.

We introduce a linear constraint on the cardinalities of vertex sets $|X_i|$.
A \emph{global linear cardinality constraint} is
an $s$-ary relation $R$ expressed by a linear inequality $a_1|X_1|+ \cdots + a_s|X_s| \leq b$,
where $a_i$ and $b$ are integers and the arguments $X_i$ are the free variables of $\varphi$.
In the extension of $\MSO$ introduced later,
global cardinality constraints are used as atomic formulas.

A \emph{local linear cardinality constraint} of $G$ on $\varphi$ is 
a mapping $\alpha^G_i \colon V(G) \to 2^{[n]}$, where
$\alpha^G_i(v) = [l_i^v, u_i^v]$ with some integers $l_i^v$ and $u_i^v$.
Each $\alpha_{i}^G$ is a constraint on the number of neighbors
of each vertex that are in $X_i$.
We say that an assignment $\vect{X}^G = (X_1^G, \ldots X_s^G)$ 
\emph{obeys} a tuple $\vect{\alpha}^G = (\alpha_{1}^{G}, \ldots, \alpha_{s}^{G})$ of
local linear cardinality constraints
if $|X_i \cap N(v)| \in \alpha_i^G(v)$ for all $v \in V(G)$ and $i \in [s]$.

An \emph{$\MSOGLLin$ formula} on a $p$-colored graph $\mathcal{G} = (G, \cfun)$ is a tuple
$(\varphi, \vect{R}, \vect{\alpha}^G)$
where $\varphi$, $\vect{R}$, and $\vect{\alpha}^G$ are defined as follows.
The tuple
$\vect{R} = (R_1, \ldots, R_g)$ is a tuple of global linear cardinality constraints, and
$\vect{\alpha}^G = (\alpha_1^G, \ldots, \alpha_s^G)$ is a tuple of local linear cardinality constraints.
The formula $\varphi$ is an $\MSO$ formula with $s$ free set variables that 
additionally has the $g$ global linear cardinality constraints $R_i$ as symbols.
Now we write $(\mathcal{G}, \vect{R}, \vect{X}^G) \models \varphi$ if
$(\mathcal{G}, \vect{X}^G) \models \varphi'$ where
$\varphi'$ is an ordinary $\MSO$ formula obtained from $\varphi$
by replacing every symbol $R_i$ with the symbol $\true$ or $\false$ representing
 the truth value of the formula $(|X^G_1|, \ldots ,|X^G_s|)\in R_i$. 

Our problem, \pname{$\MSOGLLin$ Model Checking}, is defined as follows.
\begin{myproblem}
  \problemtitle{\pname{$\MSOGLLin$ Model Checking}}
  \probleminput{A $p$-colored graph $\mathcal{G} = (G, \cfun)$ and an $\MSOGLLin$ formula $(\varphi, \vect{R}, \vect{\alpha}^G)$.}
  \problemquestiontitle{Question}
  \problemquestion{Is there an assignment $\vect{X}^G = (X_1^G, \ldots, X_s^G)$ of $\varphi$ such that 
  $(\mathcal{G}, \vect{R}, \vect{X}^G) \models \varphi$ and $\vect{X}^G$ obeys $\vect{\alpha}^G$?}
\end{myproblem}
It is known that 
\pname{$\MSOGLLin$ Model Checking} is fixed-parameter tractable parameterized by neighborhood diversity~\cite{KnopKMT19},
W[1]-hard parameterized by treedepth and feedback vertex set~\cite{KnopKMT19},
and W[1]-hard parameterized by twin-cover~\cite{KnopMT19}.

%%% Algorithm 
% -*- coding:utf-8 -*-

\section{Model checking algorithm}
\label{sec:algo}
In this section, we present our main result,
the fixed-parameter algorithm for \pname{$\MSOGLLin$ Model Checking} parameterized by vertex integrity.
Before going into the details, let us sketch the rough and intuitive ideas of the algorithm.
Recall that our goal is to find a tuple of vertex sets in a graph of bounded vertex integrity that satisfies 
\begin{itemize}
  \item an $\MSO$ formula $\varphi$ equipped with global linear cardinality constraints, and
  \item local linear cardinality constraints.
\end{itemize}

We first show that for the ordinary \pname{$\MSO$ Model Checking} with a fixed formula on graphs of bounded vertex integrity,
there is a small number of equivalence classes, called \emph{shapes}, of tuples of vertex sets
such that two tuples of the same shape are equivalent under the formula.
To use the concept of shapes,
we remove the global constraints from $\varphi$
by replacing each of them with a guessed truth value 
% \notsure{and promising that the solution we are going to find will meet the guesses we made.}
and we find a solution that meets the guesses.
Let $\varphi'$ be the resultant (ordinary) $\MSO$ formula.
We guess the shape of the solution and check whether a tuple with the guessed shape satisfies $\varphi'$
using known efficient algorithms.
If the guessed shape passed this test,
then we check whether there is a tuple with the shape satisfying the global and local cardinality constraints.
We can do this by expressing the rest of the problem as an integer linear programming (ILP) formula
as often done for similar problems (see e.g., \cite{KnopKMT19}).
The ILP formula we construct has constraints for forcing a solution to be found
\begin{itemize}
  \item to have the guessed shape,
  \item to satisfy the guessed global cardinality constraints, and
  \item to satisfy the local cardinality constraints.
\end{itemize}
The first two will be straightforward from the definitions given below.
For the local cardinality constraints,
we observe that after guessing the intersections of a $\vi(k)$-set $S$ and each set in the solution,
we know whether all vertices in $V(G) - S$ obeys the local cardinality constraints.
Thus we only need to express the local cardinality constraints in ILP for the vertices in $S$.
Finally, we will observe that the number of variables and constraints in the constructed ILP formula
depends only on $k$ and $|\varphi|$. This will give us the desired result.

\medskip

In the next subsections, we formally describe and prove the ideas explained above.

\subsection{MSO model checking}

Let $\mathcal{G} = (G, \cfun)$ be a $p$-colored graph and $S$ be a subset of $V(G)$.
We define an equivalence relation of the components of $G-S$ as follows.
Two components $A_1$ and $A_2$ of $G-S$ have the same \emph{$(\mathcal{G}, S)$-type}
if there is an isomorphism $\psi$ from $(G[S\cup A_1], \cfun|_{S\cup A_1})$ to
$(G[S\cup A_2], \cfun|_{S\cup A_2})$ 
such that the restriction $\psi|_S$ is the identity function.
We call such an isomorphism a \emph{$(\mathcal{G}, S)$-type isomorphism}.
Clearly, having the same type is an equivalence relation.
We say that a component $A$ of $G - S$ is of
$(\mathcal{G},S)$-type $t$ (or just type $t$) by using a canonical form $t$ of the members of the
$(\mathcal{G},S)$-type equivalence class of $A$. 
Denote by $\tau_{(\mathcal{G},S)}(A)$ the type of a component $A$ of $G-S$.
We will omit the index $(\mathcal{G}, S)$
if it is clear from the context.

We define the canonical form of a $(\mathcal{G},S)$-type
as the ``lexicographically'' smallest one in the equivalence class in some sense
(see \cite{GimaHKKO21} for such canonical forms of uncolored graphs).
If $S$ is a $\vi(k)$-set, then in time depending only on $p+k$
we can compute the canonical form of the equivalence class that 
a component $A$ of $G-S$ belongs to.
Thus we can compute (the canonical forms of) all $(\mathcal{G},S)$-types
in time $f(p+k) |G|^{O(1)}$ for some computable function $f$.
Furthermore, 
in time $f'(p+s+k) |G|^{O(1)}$ for some computable function $f'$,
we can compute all $(\mathcal{G}', S)$-types for all $\mathcal{G}'$ obtained from $\mathcal{G}$ by adding $s$ new colors;
that is, $\mathcal{G}' = (G,  \cfun \tupplus \vect{X})$ for some $\vect{X} \in (V(G))^{s}$.

The next lemma, due to Lampis and Mitsou~\cite{LampisM21}, is one of the main ingredients of our algorithm.
It basically says that in the ordinary \pname{$\MSO$ Model Checking}
we can ignore some part of a graph if it has too many parts that have the same type.
\begin{lemma}[\cite{LampisM21}]
    Let $\mathcal{G} = (G, \cfun)$ be a $p$-colored graph $G$,
    $S\subseteq V$, $A$ be a component of $G-S$, $|A| \le k$,
    and $\varphi$ be a closed $\MSO$ formula with $q$  quantifiers.
    If there are at least $2^{kq}+1$ type $\tau(A)$ components in $G-S$,
    then $(G, \cfun) \models \varphi$ if and only if $(G - V(A), \cfun') \models \varphi $,
    where $\cfun'$ is the restriction of $\cfun$ to $V(G)\setminus V(A)$.
    \label{lem:deletable}
\end{lemma}

Lemma~\ref{lem:deletable} leads to the following concept ``shape'',
which can be seen as equivalence classes of assignments.

\begin{definition}[Shape]
    Let $\mathcal{G} = (G, \cfun)$ be a $p$-colored graph, 
    $S$ be a $\vi(k)$-set of $G$,
    and
    $\varphi$ be an $\MSO$ formula with $s$ free set variables $(X_1, \ldots, X_s)$
    and $q$ quantifiers.
    Let $\mathcal{T}$ be the set of all $(\mathcal{G}, S)$-types,
    and $\mathcal{T}'$ be the set of all possible $(\mathcal{G}', S)$-types
    in $(p+s)$-colored graphs $\mathcal{G}'$ obtained from $\mathcal{G}$ by adding $s$ new colors.

    An $S$-\emph{shape} 
    is the pair $(\sigma_S, \sigma)$ of
    a function $\sigma_S \colon S \to 2^{\{X_1, \ldots, X_s\}} $
    and a function $\sigma \colon \mathcal{T}' \to [0,2^{kq}]\cup \{\top\}$.

    Let $\vect{X}^G = (X_1^G, \ldots, X_s^G)$ be an assignment of $\varphi$,
    and $\mathcal{G}' = (G, \cfun \tupplus \vect{X}^G)$.
    The \emph{$S$-shape} of $\vect{X}^G$ is $(\sigma_S, \sigma)$ if the following conditions are satisfied:
    \begin{itemize}
        \item
        for each $i \in [s]$ and $v \in S$, $X_i \in \sigma_S(v)$ if and only if $v \in X_i$;
        \item
        for each $t' \in \mathcal{T}'$,
        \[
            \sigma(t') = \begin{cases}
                c(t') & c(t') \leq 2^{kq},\\
                \top & \text{otherwise},
            \end{cases}
        \]
        where $c(t')$ is the number of $(\mathcal{G}', S)$-type $t'$ components of $G-S$.
    \end{itemize}
    \label{def:shape}
\end{definition}

Let $(\sigma_S,\sigma)$ be an $S$-shape.
If there is an assignment $\vect{X}^G$ of $\varphi$ such that the $S$-shape of $\vect{X}^G$ is $(\sigma_S, \sigma)$,
we say that the $S$-shape $(\sigma_S, \sigma)$ is \emph{valid}.

The following lemma indicates that $S$-shapes act as a sort of equivalence classes.
\begin{lemma}
    Let $\mathcal{G} = (G, \cfun)$ be a $p$-colored graph,
    $S$ be a $\vi(k)$-set of $G$, and
    $\varphi$ be an $\MSO$ formula with $s$ free set variables $(X_1, \ldots, X_s)$.
    Let $\vect{X}^G $ and $\vect{Y}^G$ be assignments of $\varphi$ such that their shapes are equal.
    Then, $(\mathcal{G}, \vect{X}^G) \models \varphi$ if and only if $(\mathcal{G}, \vect{Y}^G) \models \varphi$.
    \label{lem:shape-equiv}
\end{lemma}
\begin{proof}
    Let $\mathcal{G}_X = (G, \cfun \tupplus \vect{X}^G)$,
    $\mathcal{G}_Y = (G, \cfun \tupplus \vect{Y}^G)$, and
    $(\sigma_S, \sigma)$ be the $S$-shape of $\vect{X}^G$ (and of $\vect{Y}^{G}$).
    Then $\varphi$ can be seen as a closed $\MSO$ formula
    for the $(p+s)$-colored graphs $\mathcal{G}_X$ and $\mathcal{G}_Y$.
    Thus we can apply Lemma~\ref{lem:deletable} to $\mathcal{G}_X$, $S$, and $\varphi$,
    and obtain a graph $\mathcal{G}'_X = (G'_X, \cfun'_X)$, 
    such that $\mathcal{G}_X \models \varphi$ if and only if $\mathcal{G}'_X \models \varphi$
    and the number of each type $t$ components of $G'_X-S$
    is at most $2^{kq}$, where $q$ is the number of quantifiers in $\varphi$.
    We also obtain a graph $\mathcal{G}'_Y$ in the same way as for $\mathcal{G}'_X$.
    This reduction does not delete any vertex of $S$.
    The number of components for each type $t$ of $G'_X - S$ or $G'_Y - S$
    is $\sigma(t)$ if $\sigma(t) \neq \top$ and $2^{kq}$ if $\sigma(t) = \top$.
    Therefore, there is an isomorphism from $\mathcal{G}'_X$ to $\mathcal{G}'_Y$,
    and thus $\mathcal{G}'_X \models \varphi$ if and only if $\mathcal{G}'_Y\models \varphi$
    by Lemma~\ref{lem:iso-equiv}.
\end{proof}

%---
Now, we estimate the number of candidates for $S$-shapes.
Observe that in Definition~\ref{def:shape},
the number of candidates for $\sigma_S$
depends only on $k$ and $s$.
The size of $\mathcal{T}$
% is at most $(2^k)^p \cdot 2^{k^2}$, 
depends only on $k$ and $p$
because it is at most the product of 
the number of $k \times k$ adjacency matrices and 
the number of $p$-color lists for graphs of at most $k$ vertices.
% lists of $p$ subsets of the set of at most $k$ vertices.
Similarly, the size of $\mathcal{T}'$
depends only on $k$, $p$ and $s$.
Since $\sigma$ is a function from $\mathcal{T}'$ 
to $[0, 2^{kq}]\cup \{\top\}$, the number of candidates for $\sigma$ 
depends only on $k$, $p$, $s$, and $q$.
Thus the number of $S$-shapes depends only on $k$, $p$, $s$, and $q$.
%---
% $(2^{kq})^{2^{k^2}}$
\begin{observation}
    Let $\mathcal{G} = (G, \cfun)$ be a $p$-colored graph with $n$ vertices,
    $S$ be a $\vi(k)$-set of $G$,
    and
    $\varphi$ be an $\MSO$ formula with $s$ free set variables and $q$ quantifiers.
    The number of $S$-shapes depends only on $k$, $p$, $s$ and $q$.
    \label{obs:number-of-shape}
\end{observation}

% ==========================================================

\subsection{Pre-evaluating the global constraints}
\label{sec:alg-global}

Recall that in an $\MSOGLLin$ formula, the global cardinality constraints are used as atomic formulas.
Namely, each of them takes the value true or false depending on the cardinalities of the free variables.
To separate these constraints from the model checking process,
the approach of pre-evaluation was used in the previous studies~\cite{GanianO13,KnopKMT19}.

\begin{definition}[Pre-evaluation]
Let $\mathcal{G} = (G, \cfun)$ be a $p$-colored graph,
and $(\varphi, \vect{R}, \vect{\alpha}^G)$ be an $\MSOGLLin$ formula
where $ \vect{R} = (R_1, \ldots, R_g)$.
We call a function 
$\gamma \colon \{R_1, \ldots, R_g\} \to \{\mathrm{\true} , \mathrm{\false}\}$
a \emph{pre-evaluation}.
Denote by $\gamma(\varphi)$ the $\MSO$ formula 
that obtained by mapping each global linear cardinality constraints $R_i$ by $\gamma$.
\end{definition}

Since each global linear cardinality constraint $R_i$ can be represented
by a linear inequality,
so is its complement $\bar{R}_i = [0, n]^{s} \setminus {R_i}$.
Thus, for a pre-evaluation $\gamma$, the integers $x_1, \ldots, x_s \in [0,n]$ that 
satisfies the following conditions can be represented by a 
system of linear inequalities:
\begin{itemize}
    \item If $\gamma(R_i) = \true$, then $(x_1,\ldots, x_s) \in R_i$.
    \item Otherwise, $(x_1, \ldots, x_s) \notin R_i$.
\end{itemize}
Denote by $\vect{R}_{\gamma}(x_1, \dots, x_s)$ this system of linear inequalities.
If assignment $\vect{X}^G = (X^G_1, \dots, X^G_s)$ of $\varphi$
satisfies the system of inequalities $\vect{R}_{\gamma}(|X^G_1|, \ldots, |X^G_s|)$,
we say that $\vect{X}$ \emph{meets} the pre-evaluation $\gamma$.

% ==========================================================
\subsection{Making the local constraints uniform}
\label{sec:alg-local}
Observe that for a $\vi(k)$-set $S$ of a graph $G$,
a vertex $v$ of a component of $G-S$ has at most $k-1$ neighbors.
In other words, $|N(v)| \in [0,k-1]$ for each vertex $v \in V(G-S)$.
Therefore, $|N(v) \cap X| \in \alpha(v)$ if and only if $|N(v) \cap X| \in \alpha(v) \cap [0, k-1]$
for every combination of $X \subseteq V(G)$, $\alpha \colon V(G) \to [0,n]$, and $v \in V(G-S)$.
Thus, we can reduce the range of local constraints as follows.
\begin{observation}
    Let $\mathcal{G} = (G, \cfun)$ be a $p$-colored graph,
    $S$ be a $\vi(k)$-set of $G$,
    $(\varphi, \vect{R}, \vect{\alpha}^G) $ be an $\MSOGLLin$ formula
    where $\vect{\alpha}^{G} = (\alpha_1, \ldots, \alpha_s)$,
    and $\vect{X}^G$ be an assignment of $\varphi$.
    Denote by $\vect{\beta}^G$ the local constraints obtained from
    $\vect{\alpha}^{G}$ by
    restricting to $\alpha_i^G(v) \cap [0,k-1]$ 
    for each $v \in G-S$ and $i \in [s]$.
    Then, $\vect{X}^G$ obeys $\vect{\alpha}^G$ if and only if $\vect{X}^{G}$ obeys $\vect{\beta}^G$.
    \label{obs:localcard}
\end{observation}

\begin{definition}[Uniform]
    Let $\mathcal{G} = (G, \cfun)$ be a $p$-colored graph with $n$ vertices,
    $S$ be a $\vi(k)$-set of $G$,
    and $(\varphi, \vect{R}, \vect{\alpha}^G)$
    be an $\MSOGLLin$ formula where $\vect{\alpha}^G = (\alpha_1, \ldots, \alpha_s)$.
    We say that
    the graph $\mathcal{G}$ is \emph{uniform} on 
    the local constraints $\vect{\alpha}^G$
    if for every pair of components $A_1$ and $A_2$ of $G-S$ with the same $S$-shape,
    there is a $(\mathcal{G}, S)$-type isomorphism $\psi$ from $A_1$ to $A_2$
    such that $\alpha^G_i(v) = \alpha^G_i(\psi(v))$ for all $i \in [s]$ and $v \in V(A_1)$.
    % (i.e. every vertex of the components has the same local constraints).
\end{definition}
We can obtain a uniform graph $\mathcal{G}'$ on $\vect{\alpha}^G$
from nonuniform graph $\mathcal{G} = (G, \cfun)$ on $\vect{\alpha}^G$
as follows.
For each $i \in [s]$,
assign to every vertex $v \in V(G)\setminus S$ new colors $C^i_{\alpha_i(v)}$
 corresponding to the local constraints $\alpha_i(v)$.
Then we obtain a uniform graph 
$\mathcal{G}' = (G, \vect{C} \tupplus (C^i_B)_{i\in [s], B \subseteq [0, k-1]})$.
The number of new colors added to $\mathcal{G}'$ is at most $sk^2$.

\begin{lemma}
    Let $\mathcal{G} = (G, \cfun)$ be a $p$-colored graph with $n$ vertices,
    $S$ be a $\vi(k)$-set of $G$,
    and $(\varphi, \vect{R}, \vect{\alpha}^G )$ be an $\MSOGLLin$ formula
    where $\vect{\alpha}^G= (\alpha_1, \ldots, \alpha_s)$.
    Assume that the graph $\mathcal{G}$ is \emph{uniform} on 
    the local constraints $\vect{\alpha}^G$.
    Let $\vect{X}^G = (X_1^G, \ldots, X_s^G)$ and $\vect{Y}^G = (Y_1^G, \ldots, Y_s^G)$
    be assignments of $\varphi$ with the same $S$-shape $(\sigma_S, \sigma)$.
    Then for each $i \in [s]$ and $v \in V(G) \setminus S$,
    $|X_i^G \cap N(v)| \in \alpha^G_i(v)$ 
    if and only if
    $|Y_i^G \cap N(v)| \in \alpha^G_i(v)$.
    \label{lem:uniform-equiv}
\end{lemma}
\begin{proof}
    By symmetry, it suffices to prove the only-if direction.
    Assume that $|X_i^G \cap N(v)| \in \alpha^G_i(v)$ for each $i \in [s]$ and $v \in V(G) \setminus S$.
    Let $\mathcal{G}_X = (G, \cfun \tupplus \vect{X}^G)$,
    $\mathcal{G}_Y = (G, \cfun \tupplus \vect{Y}^G)$, 
    and $A_Y$ be a component of $G-S$.
    Since the $S$-shape of $\vect{X}^G$ and $\vect{Y}^G$ are the same,
    there is a component $A_X$ of $G-S$ such that the $(\mathcal{G}_X, S)$-type of $A_X$ is 
    equal to the $(\mathcal{G}_Y, S)$-type of $A_Y$.
    Then, there is an isomorphism $\psi$ from $A_Y$ to $A_X$ such that
    $|Y_i^G \cap N(v)| = |X_i^G \cap N(\psi(v))|$
    and $\alpha_i^G(v) = \alpha_i^G(\psi(v))$
    for each $v \in A_Y$ and $i \in [s]$,
    because 
     $\mathcal{G}$ is uniform on $\vect{\alpha}$.
    Therefore, $|Y_i^G \cap N(v)| = |X_i^G \cap N(\psi(v))| \in \alpha^G_i(\psi(v)) = \alpha^G_i(v)$
    for each $v \in V(G)\setminus S$ and $i \in [s]$.
\end{proof}

\subsection{The whole algorithm}
We reduce the feasibility test of global and local constraints 
to the feasibility test of an ILP formula with a small number of variables.
The variant of ILP we consider is formalized as follows.
% The problem is defined as follows:
\begin{myproblem}
    \problemtitle{\textsc{$p$-Variable Integer Linear Programming Feasibility ($p$-ILP)}}
    \probleminput{A matrix $A \in \mathbb{Z}^{m\times p}$ and a vector $\vect{b} \in \mathbb{Z}^m$.}
    \problemquestiontitle{Question}
    \problemquestion{Is there a vector $\vect{x} \in \mathbb{Z}^p$ such that $A \vect{x} \leq \vect{b}$?}
\end{myproblem}
Lenstra~\cite{Lenstra83} showed that $p$-ILP is fixed-parameter tractable parameterized by the number of variables $p$, 
and this algorithmic result was later improved by Kannan~\cite{Kannan87} and by Frank and Tardos~\cite{FrankT87}.
\begin{theorem}[\cite{Lenstra83,Kannan87,FrankT87}]
    \textsc{$p$-ILP} can be solved using $O(p^{2.5p+o(p)}\cdot L)$
    arithmetic operations and space polynomial in $L$, where $L$ is the number of the bits in the input.
    \label{thm:ILP}
\end{theorem}
The next technical lemma is the main tool for our algorithm.

\begin{lemma}
    Let $\mathcal{G} = (G, \cfun)$ be a $p$-colored graph,
    $S$ be a $\vi(k)$-set of $G$,
    and $(\varphi, \vect{R}, \vect{\alpha}^G)$ be an $\MSOGLLin$ formula
    where $\varphi$ has $s$ free set variables $X_1, \ldots X_s$,
    $\vect{R}= (R_1, \ldots, R_g)$, and $\vect{\alpha}^G = (\alpha_1, \ldots, \alpha_s)$.
    Assume that $\mathcal{G}$ is uniform on $\vect{\alpha}^G$.
    Then, there is an algorithm that given a valid $S$-shape $(\sigma_S, \sigma)$,
    decides whether there exists an assignment $\vect{X}^G = (X_1^G, \ldots, X_s^G)$
    such that its $S$-shape is $(\sigma_S, \sigma)$,
    $(\mathcal{G},\vect{X}^G, \vect{R}) \models \varphi$, 
    and $\vect{X}^G$ obeys $\vect{\alpha}^G$
    in time $f(k, |\varphi|) n^{O(1)}$ for some computable function $f$.
    \label{lem:shape-eval}
\end{lemma}
\begin{proof}

    Our task is to find an assignment $\vect{X}^G$
    such that 
    \begin{enumerate}
        \item the $S$-shape of $\vect{X}^G$ is $(\sigma_S, \sigma)$, \label{itm:shape}
        \item $(\mathcal{G},\vect{X}^G, \vect{R}) \models \varphi$, and \label{itm:global}
        \item $|X_i^G \cap N(v)| \in \alpha_i^G(v)$ for all $v \in V(G)$ and $i \in [s]$.\label{itm:local}
    \end{enumerate}
    Condition~\ref{itm:shape} can be handled easily by linear inequalities in our ILP formulation.
    Condition~\ref{itm:global} is equivalent to the condition that
    there exists a pre-evaluation~$\gamma$ 
    such that $(\mathcal{G},\vect{X}^G) \models \gamma(\varphi)$,
    and $\vect{X}^{G}$ meets $\gamma$.
    We check whether $(\mathcal{G},\vect{X}^G) \models \gamma(\varphi)$ 
    and whether $\vect{X}^{G}$ meets $\gamma$ separately.
    Furthermore, Condition~\ref{itm:local} is checked separately for vertices in $S$ 
    and for vertices in $V(G) \setminus S$.

    \paragraph*{Step 1. Guessing and evaluating a pre-evaluation for the global constraints.}
    We guess a pre-evaluation $\gamma$
    from $2^g \leq 2^{|\varphi|}$ candidates.
    We check whether each shape-$(\sigma_S, \sigma)$ assignment $\vect{X}^G$
    satisfies the $\MSO$ formula $\gamma(\varphi)$,
    i.e., $(\mathcal{G}, \vect{X}^G) \models \gamma(\varphi)$.
    By Lemma~\ref{lem:shape-equiv},
    we only need to check whether $(\mathcal{G}, \vect{X}^G) \models \gamma(\varphi)$
    for an arbitrary assignment $\vect{X}^G$ of $S$-shape $(\sigma_S, \sigma)$.
    This can be done in $f(k, |\varphi|) n^{O(1)}$ time~\cite{Courcelle90mso1,LampisM21}.
    Note that even if $(\mathcal{G}, \vect{X}^G) \models \gamma(\varphi)$ is true, 
    this arbitrarily chosen $\vect{X}^G$ may not meet $\gamma$.
    In Step~3, we find a shape-$(\sigma_S, \sigma)$ assignment that meets $\gamma$. 

    \paragraph*{Step 2. Checking the local constraints for the vertices in $V(G) \setminus S$.}
    By Lemma~\ref{lem:uniform-equiv}, 
    we can check whether all shape-$(\sigma_S, \sigma)$ assignments 
    satisfy the local constraints for the vertices in $V(G) \setminus S$
    by constructing an arbitrary assignment $\vect{Y}^G$ of $S$-shape $(\sigma_S, \sigma)$
    and testing whether
    $|Y_i^G \cap N(v)| \in \alpha_i$ for all $v \in V(G) \setminus S$ and $i\in [s]$.
    Since constructing an assignment $\vect{Y}^G$ can be done in $f(k, |\varphi|)n^{O(1)}$ time,
     this test can be done in $f(k, |\varphi|)n^{O(1)}$ time.
    
    \paragraph*{Step 3. Constructing a system of linear inequalities for the remaining constraints.}
    By Steps~1 and 2, it suffices to check whether there exists an assignment 
    $\vect{X}^G = (X_1^G, \dots, X_s^G)$ that satisfies the following conditions:
    % From now on, we express the remaining constraints as a system of inequalities.
    % First, we express 
    \begin{enumerate}
        \item the $S$-shape of $\vect{X}^G$ is $(\sigma_S, \sigma)$, \label{itm2:shape}
        \item $\vect{X}^G$ meets the pre-evaluation $\gamma$, \label{itm2:pre} and
        \item $\vect{X}^G$ obeys the local constraints $\vect{\alpha}^G$ for the vertices in $S$.
    \end{enumerate}
    To this end,
    we construct a system of linear inequalities as follows.

    In the following, we denote by $\mathcal{G}'$ the $(p+s)$-colored graph $(G, \cfun \tupplus \vect{X}^G)$,
    where $\vect{X}^G$ is a hypothetical solution we are searching for.

    Let $\mathcal{T}$ be the set of all $(\mathcal{G}, S)$-types.
    For every $t \in \mathcal{T}$, 
    the number of type-$t$ components of $G-S$ is denoted by $n_t$.
    Let $\mathcal{T}'$ be the set of all possible $(\mathcal{H}, S)$-types in $(p+s)$-colored graphs
    $\mathcal{H}$ obtained from $\mathcal{G}$ by adding $s$ new colors.
    Observe that $\mathcal{T}'$ is a superset of the set of all $(\mathcal{G}', S)$-types,
    no matter how $\vect{X}^G$ is chosen.
    \newcommand{\invt}[1]{#1|_p}
    For every $t' \in \mathcal{T}'$,
    the $(\mathcal{G}, S)$-type of a type-$t'$ component is uniquely determined 
    and is denoted by $\invt{t'}$.
    (This notation comes from the fact that the $(\mathcal{G}, S)$-type of a type-$t'$ component 
    can be determined by considering the first $p$-colors.)
    For every $t' \in \mathcal{T}'$,
    we introduce the variable $x_{t'}$
    that represents the number of $(\mathcal{G}', S)$-type $t'$ components.
    The condition that the variables $x_{t'}$ agree with $\sigma$ can be expressed as follows:
    \begin{align*}
        \sum_{t' \in \mathcal{T}',\; \invt{t'}= t} x_{t'} &= n_t & &\text{for every } t \in \mathcal{T},\\
        x_{t'}  &= \sigma(t') & &\text{for every } t' \in \mathcal{T}'\ \text{such that}\ \sigma(t') \neq \top, \\
        x_{t'}  &\geq 2^{kq}+1 & &\text{for every } t' \in \mathcal{T}'\ \text{such that}\ \sigma(t') =\top.
    \end{align*}
    For every $i \in [s]$,
    we introduce an auxiliary variable $y_i$ that represents the size of the set $X^G_i$,
    which is determined by the variables $x_{t'}$.
    The variables $y_i$ can be expressed as follows:
    \begin{align*}
        y_i &= |\{v \in S \mid X_i \in \sigma_S(v)\}| + \sum_{t' \in \mathcal{T}'} x_{t'} \cdot \#_i (x_{t'}) & & \text{for every } i \in [s],
    \end{align*}
    where $\#_i (x_{t'})$ is the number of vertices with color $p+i$ in a type-$t'$ component, i.e., 
    the number of vertices assigned to the variable $X_i$ in a type-$t'$ component.
    Then, as mentioned in Section~\ref{sec:alg-global},
    the global constraints that match the pre-evaluation $\gamma$ can be represented by the system of inequalities $R_{\gamma}(y_1, \ldots, y_s)$.

    Finally, we formulate the local constraints for the vertices in $S$ into a system of inequalities.
    For every $v \in S$, $i \in [s]$, and $t' \in \mathcal{T}'$,
    the number of neighbors of $v$ with color $p+i$ (i.e., in the set variable $X_i$) 
     in a type-$t'$ component is denoted by $d_{i,t'}(v)$
    (i.e., $d_{i,t'}(u) = |N(u) \cap X_i \cap V(A)|$ where $A$ is a type-$t'$ component).
    All constants $d_{i,t'}(v)$ can be computed in $f(k, |\varphi|)n^{O(1)}$ time.
    For every $i \in [s]$ and $v \in S$, we introduce an auxiliary variable $z_{v, i}$ that represents 
    the number of neighbors of $v$ in the set $X_i$, which is determined by the variables $x_{t'}$.
    The variables $z_{v, i}$ can be expressed as follows:
    \begin{equation*}
        z_{v,i} = |\{u \in N(v) \cap S \mid X_i \in \sigma_S(u) \} | + \sum_{t' \in \mathcal{T}'} d_{i,t'}(v) x_{t'}
        \quad \text{for every } v\in S,\ i\in[s].
    \end{equation*}
    Since the local constraints $\alpha_i^G$ can be expressed by $\alpha_i^G(v) = [l_i^v, u_i^v]$ with some integers
    $l_i^v$ and $u_i^v$ for every vertex $v$,
    the local constraints for vertices in $S$ can be expressed as follows:
    \begin{align*}
        l_i^v \leq z_{v,i}  \leq u_i^v & & \text{for every } v \in S, i \in [s].
    \end{align*}

    By finding a feasible solution to the ILP formula constructed above,
    we can find a desired assignment $\vect{X}^G$.
    Since the number of the variables in the ILP formula depends only on $k$ and $|\varphi|$, the lemma follows
    by Theorem~\ref{thm:ILP}.
\end{proof}

\begin{theorem}
    % Let $\varphi$ be an $\MSOGLLin$ formula with $s$ free set variables $X_1, \ldots, X_s$
    % and $g$ global constraint symbols.
    % Let $\mathcal{G} = (G, L, \cfun)$ be a $(p,q)$-graph,
    % % a $p$-label $L$, and a $q$-color $\cfun$,
    % $S$ be a $\vi(k)$-set of $G$,
    % $\mathcal{R}^G = R^G_1, \ldots R^G_g$ be linear global constraints,
    % and $\mathcal{B}^G = (\beta^G_1, \beta^G_2, \ldots, \beta^G_s)$ be
    % a linear local constraint.
    % There exists an algorithm that decides
    % whether there exists an assignment $\mu \colon \{X_1, \ldots, X_s\} \to 2^{V(G)}$
    % such that $\mathcal{G},\mu, \mathcal{R}^G \models \varphi$ and 
    % $\mu$ obeys $\mathcal{B}^G$
    % in time $f(k, |\varphi|) n^{O(1)}$ for some computable function $f$.
    \textsc{$\MSOGLLin$ Model Checking} is fixed-parameter tractable parameterized by 
    $\vi(G)$ and $|\varphi|$.
    \label{thm:msogl-alg}
\end{theorem}

\begin{proof}
    Let $k=\vi(G)$. Let $S$ be a $\vi(k)$-set.
    Such a set can be found in $O(k^{k+1}n)$ time~\cite{DrangeDH16}.
    We construct a uniform graph $\mathcal{H} = (G, \cfun')$ on $\vect{\alpha}^G$
     from the input graph $\mathcal{G} = (G, \cfun)$ 
    as described in Section~\ref{sec:alg-local}.
    Here, the number of colors of $\mathcal{H}$ depends only on $k$, $p$, and $s$.
    We compute the $(\mathcal{H}, S)$-types of the components of $G-S$
    and count the number of $(\mathcal{H}, S)$-type $t$ components for each $t$.
    This can be done in $f(k, |\varphi|)n$ time with some computable function~$f$.

    We guess an $S$-shape $(\sigma_S, \sigma)$ of an assignment of the input formula $\varphi$.
    By Observation~\ref{obs:number-of-shape}, 
    the number of candidates for $(\sigma_S, \sigma)$ depends only on $k$, $p$, and $s$.
    We check whether the guess shape $(\sigma_S, \sigma)$ is valid.
    This can be done by checking whether $(\sigma_S, \sigma)$ is consistent with 
    the number of components of all $(\mathcal{H}, S)$-types.
    Hence, this can be done in $f(k, |\varphi|)n$ time with some computable function $f$.

    By Lemma~\ref{lem:shape-eval} with the graph $\mathcal{H}$, $\vi(k)$-set $S$,
    and the input $\MSOGLLin$ formula $(\varphi, \vect{R}, \alpha^G)$,
    the theorem follows.
\end{proof}

%%% Extension to MSO2 
% -*- coding:utf-8 -*-

\section{Extension to $\MSO_{2}$}
\label{sec:ext}

% (This is a full version of \cref{sec:short-ext}.) % todo: to be removed in the full version

In this section, we consider an extension of $\MSOGLLin$ to $\MSO_{2}$.
The $\MSO_{2}$ logic\footnote{%
The $\MSO_{2}$ logic is also known as the $\GSO$ logic, which stands for guarded second-order logic.} 
on graphs is a generalization of $\MSO$ ($=\MSO_{1}$) that additionally allows edge variables, edge-set variables,
and an atomic formula $I(x,y)$ meaning that the edge assigned to $y$ is incident to the vertex assigned to  $x$.
It is known that $\MSO_{2}$ is strictly stronger than $\MSO_{1}$ for general graphs
in the sense that there are some properties that can be expressed in $\MSO_{2}$ but not in $\MSO_{1}$
(e.g., Hamiltonicity~\cite{CourcelleE12}).
On the other hand, for graphs of bounded treewidth,
the model checking problem for $\MSO_{2}$ can be reduced to the one for $\MSO_{1}$ in polynomial time (see e.g.,~\cite{CourcelleE12}).
Using a similar reduction, 
we show that the same holds for $\MSOGLLin$ on graphs of bounded vertex integrity.

Now we define the extension of $\MSOGLLin$ with $\MSO_{2}$, which we call $\GSOGLLin$.
In $\GSOGLLin$, 
the local cardinality constraints for vertex-set variables
and the global cardinality constraints work in the same way as in $\MSOGLLin$.
The local cardinality constraints for an edge-set variable $X$ at a vertex $v$
restricts the number of edges in $X$ incident to $v$.
A \emph{$\GSOGLLin$ formula} on a $p$-colored graph $\mathcal{G} = (G, \cfun)$ is a tuple $(\varphi, \vect{R}, \vect{\alpha}^G)$,
where $\vect{R} = (R_1, \ldots, R_g)$ and $\vect{\alpha}^G = (\alpha_1^G, \ldots, \alpha_s^G)$ are the global and local cardinality constraints,
and $\varphi$ is an $\MSO_{2}$ formula with $s$ free set variables that 
additionally equipped with symbols $R_{1}, \dots, R_{g}$.
The problem \pname{$\GSOGLLin$ Model Checking} is formalized as follows.
\begin{myproblem}
  \problemtitle{\pname{$\GSOGLLin$ Model Checking}}
  \probleminput{A $p$-colored graph $\mathcal{G} = (G, \cfun)$, and a $\GSOGLLin$ formula $(\varphi, \vect{R}, \vect{\alpha}^G)$.}
  \problemquestiontitle{Question}
  \problemquestion{Is there an assignment $\vect{X}^G = (X_1^G, \ldots, X_s^G)$ of $\varphi$ such that 
  $(\mathcal{G}, \vect{R}, \vect{X}^G) \models \varphi$ and $\vect{X}^G$ obeys $\vect{\alpha}^G$?}
\end{myproblem}

In the rest of this section, we show the following theorem.
\begin{theorem}
\label{thm:gsogl-alg} % todo: uncomment this line in the full version
\pname{$\GSOGLLin$ Model Checking} is fixed-parameter tractable parameterized by 
$\vi(G)$ and $|\varphi|$.
\end{theorem}
By Theorem~\ref{thm:msogl-alg}, it suffices to present a polynomial-time reduction 
from \pname{$\GSOGLLin$ Model Checking} to \pname{$\MSOGLLin$ Model Checking}
that does not increase the vertex integrity too much.
Given an instance of \pname{$\GSOGLLin$ Model Checking} that consists of
a $p$-colored graph $\mathcal{G} = (G=(V,E), \cfun)$ and a $\GSOGLLin$ formula $(\varphi, \vect{R}, \vect{\alpha}^G)$,
we construct an equivalent instance of \pname{$\MSOGLLin$ Model Checking}.
Since most of the reduction below is rather standard,
we basically give the construction only, and add some remarks for important points.

\paragraph*{Modifying the graph.}
We construct a $(p+1)$-colored graph $\mathcal{G}' = (G', \cfun')$ from $\mathcal{G}$
by adding a new color consists of new vertices $v_{e}$ for all $e \in E$,
and adding the edges between $v_{e}$ and the endpoints of $e$ for each $e \in E$.
That is, 
$\cfun' = \cfun \tupplus (C_{E})$ with $C_{E} = \{v_{e} \mid e \in E\}$,
$V(G') = V \cup C_{E}$, and 
$E(G') = E \cup \{\{u,v_{e}\}, \{w,v_{e}\} \mid e = \{u,w\} \in E\}$.
Note that we do not forget the original edge set $E$ (which is redundant)
because handling local cardinality constraints for vertex-set variables is simpler with $E$.
We can easily see that $\vi(G') \le (\vi(G))^{2}$.
\begin{observation}
\label{obs:subdiv}
If $\vi(G) = k$, then $\vi(G') \le k^{2}$.
\end{observation}
\begin{proof}
Let $S$ be a $\vi(k)$-set of $G = (V,E)$. We show that $S$ is a $\vi(k^{2})$-set of $G'$.
For $e \in E(G[S])$, the vertex $v_{e}$ forms a singleton component in $G'-S$.
Let $D$ be a component of $G'-S$ that contains at least one vertex in $V$.
By the definition of $G'$, 
there exists a component $C$ of $G-S$ such that
\[
  V(D) 
  = 
  V(C) 
  \cup \{v_{e} \mid e \in E(G[C])\} 
  \cup \{v_{\{u,w\}} \mid \{u,w\} \in E, u \in V(C), w \in S\}.
\]
This implies that 
$|S \cup V(D)| \le |S \cup V(C)| + \binom{|V(C)|}{2}+ |S| \cdot |V(C)| \le k^{2}$.
\end{proof}

\paragraph*{Modifying the formula.}
From $\varphi$, we obtain a formula $\varphi'$ as follows.
All edge variables and edge-set variables in $\varphi$ are interpreted
as vertex variables and vertex-set variables, respectively, with the same names in $\varphi'$.
The formula $\varphi'$ asks that all ex-edge variables and ex-edge-set variables are taken from $C_{E}$,
and none of the other variables intersect $C_{E}$. 
For example, if $Y$ is an edge-set variable in $\varphi$, 
we replace the maximal subformula $\psi_{Y}$ of $\varphi$ where $Y$ is defined in with 
$\psi_{Y} \land (\forall y.(\lnot(y \in Y) \lor y \in C_{E}))$, which means $\psi_{Y} \land (Y \subseteq C_{E})$.
Finally, we replace each $I(x,y)$ in $\varphi$ with $E(x,y)$.

\paragraph*{Modifying the cardinality constraints.}
We keep the original local and global cardinality constraints.
Clearly, the global cardinality constraints work as before since nothing changed for them.
Observe that the new vertices in $C_{E}$ do not have local cardinality constraints.
For them, we add dummy local constraints that restrict nothing (e.g., $[0,|V|]$).
Let $X$ be a set variable in $\varphi'$.
Recall that $\varphi'$ ensures that $X \subseteq C_{E}$ if $X$ is an edge-set variable in $\varphi$, 
and $X \cap C_{E} = \emptyset$ otherwise.
Recall also that $G'$ keeps the original edges in $E$.
Therefore, the local cardinality constraints for both ex-edge-set variables and ex-vertex-set variables work correctly.

%%% Conclusion
\section{Concluding remarks}

In this paper, we obtained an algorithmic meta-theorem for graphs of bounded vertex integrity
in a framework introduced as an extension of $\MSO$ by Knop, Kouteck\'{y}, Masa\v{r}\'{\i}k, and Toufar~\cite{KnopKMT19}.
Namely, we showed that \pname{$\MSOGLLin$ Model Checking} (or more generally, \pname{$\GSOGLLin$ Model Checking})
is fixed-parameter tractable parameterized by vertex integrity.
This result partially covers the results of the previous study~\cite{GimaHKKO21}:
some problems admit direct translations from their definitions to expressions in $\MSOGLLin$
(e.g., \pname{Equitable $r$-Coloring}) and 
some need non-trivial modifications to make them expressible in $\MSOGLLin$
 (e.g., {\pname{Capacitated Vertex Cover}}).
For some other problems (e.g., \pname{Imbalance} and \pname{Max Common Subgraph}),
we were not able to determine that they can be captured by our framework or not.
Also, the result newly gives algorithms for \pname{Fair Evaluation Problems}~\cite{KnopMT19}.
It would be interesting to ask whether
there is a meta-theorem that can be applied to a larger class of problems parameterized by vertex integrity.
(See \cref{sec:app}.)

We may also consider the fine-grained complexity of our problem.
We did not explicitly state the time complexity of our fixed-parameter algorithms.
If we carefully analyze the running time using the algorithm by Lampis and Mitsou~\cite{LampisM21},
then we can show that the algorithms run in time triple exponential in a polynomial function of the parameter.
For the ordinary \pname{$\MSO$ Model Checking},
it is known that under $\ETH$,
 there is no $2^{2^{o({k^2})}}n^{O(1)}$-time algorithm,
where $k$ is the vertex integrity of the input graph $G$
and $n$ is the number of vertices of $G$~\cite{LampisM21}.
This double-exponential lower bound applies also to our generalized problem. 
Filling this gap would be an interesting challenge.

%%% References
%\bibliography{msovi}
% \putbib 
% \end{bibunit}

\clearpage

\appendix

% %%% Extension to MSO2 (full version)
% \input{ext}

%%% Applications
% -*- coding:utf-8 -*-

\section{Applications of the meta-theorem}
\label{sec:app}

Here we present some applications of our main result (\cref{thm:msogl-alg,thm:gsogl-alg}).
We first show that the theorems give some new examples that are fixed-parameter tractable parameterized by vertex integrity.
We also observe that some known results can be obtained from the theorems.
We finally add some remarks on known results that are not captured by the theorems.
In the following, we assume that the readers are familiar with $\MSO$ expressions of graph problems and
omit the actual $\MSO$ formulas. 
See \cite[Section~7.4]{CyganFKLMPPS15} for examples of $\MSO$ expressions of some basic graph properties.

\subsection{New results obtained by the meta-theorem}

\subsubsection{\pname{Fair $\MSO$ Evaluation}}

One of the most immediate consequences of our result is
the fixed-parameter tractability of \pname{Fair $\MSO_{2}$ Evaluation} parameterized 
by vertex integrity and the length of the input $\MSO_{2}$ formula.
\pname{Fair $\MSO_{2}$ Evaluation} asks the existence of a tuple of sets of vertices or edges
such that the tuple satisfies a given $\MSO_{2}$ formula 
and for each set in the tuple, there is an upper bound of the number of vertices or edges
that each vertex in the graph can be adjacent to or incident to.
For example, \pname{Fair Vertex Cover} asks the existence of a vertex cover $C$
with a condition that each vertex in the graph has at most $k$ neighbors in $C$.
It is known that \pname{Fair Vertex Cover} is W[1]-hard parameterized
by both treedepth and feedback vertex set number~\cite{KnopMT19}.
Clearly, \pname{Fair $\MSO_{2}$ Evaluation} is a special case of \pname{$\GSOGLLin$ Model Checking}.
The fairness of solutions in this sense was first introduced for specific problems~\cite{LinS89a}
and later studied in general settings~\cite{KnopMT19,MasarikT20}.

\pname{Defective Coloring}~\cite{CowenCW86} asks, given a graph $G = (V,E)$ and integers $k$ and $d$,
whether $V$ can be partitioned into $k$ sets $V_{1}, \dots, V_{k}$ 
such that the maximum degree of $G[V_{i}]$ is at most $d$ for each $i \in [k]$.
\pname{Defective Coloring} parameterized by treedepth is W[1]-hard
for every fixed $k \ge 2$~\cite{BelmonteLM20}.
As observed in~\cite{KnopMT19}, the problem is 
equivalent to the one that asks for an edge set $F \subseteq E$
such that $G[F]$ has maximum degree at most $d$ and $G-F$ admits a proper $k$-coloring.
Hence, this is a typical example of \pname{Fair $\MSO_{2}$ Evaluation}
with a formula length depending only on $k$.
Observe that if $k \ge \vi(G)$, then $\langle G, k, d \rangle$ is a yes-instance of \pname{Defective Coloring}
as $G$ admits a proper $\vi(G)$-coloring. Thus we have the following result.
\begin{theorem}
\pname{Defective Coloring} is fixed-parameter tractable parameterized by vertex integrity.
\end{theorem}

\subsubsection{Alliances in graphs}

A nonempty set $S \subseteq V$ is a \emph{defensive alliance} of a graph $G = (V,E)$ if
\begin{equation}
  |N[v] \cap S| \ge |N[v] \setminus S| \ \text{for every} \ v \in S. \label{eq:alliance}
\end{equation}
Intuitively, a defensive alliance is a set of vertices that is ``safe'' under attacks from its neighborhood~\cite{KristiansenHH04}.
The problem \pname{Defensive Alliance} asks, given a graph $G$ and an integer $k$,
whether $G$ contains a defensive alliance of size at most $k$.
\pname{Defensive Alliance} is W[1]-hard parameterized by treewidth~\cite{BliemW18}.

Now let us express \pname{Defensive Alliance} as \pname{$\GSOGLLin$ Model Checking}.
We first subdivide each edge in the graph. We call the new vertices introduced as $W$ and the new graph as $H = (V \cup W, F)$.%
\footnote{To be more precise, we color the vertices in $W$ with a new color to distinguish them from the original vertices.}
We can see that $\vi(H) \le (\vi(G))^{2}$ as the graph $H$ here is a subgraph of the graph $G'$ in \cref{obs:subdiv}.
Let $X$ be a subset of $V$ and $Y$ a subset of $F$.
We can express the following conditions as a $\GSOGLLin$ formula with free variables corresponding to $X$ and $Y$:
\begin{itemize}
  \item $|X| \le k$;

  \item $Y$ is the set of edges $e$ such that $e$ has one endpoint in $X$
  and the other endpoint of $e$, which belongs to $W$, is adjacent to a vertex in $V \setminus X$;

  \item each $v \in V$ is incident to at most $|N_{G}[v]|/2$ edges in $Y$;
\end{itemize}
To see the correctness, observe that $Y$ corresponds to the set of edges in the original graph $G$
between the vertices in $X$ and $V \setminus X$.
Then the definition of defensive alliances ask the third condition.

Several variants and generalizations of defensive alliances are studied~\cite{KristiansenHH04,FrickeLHHH03,CamiBDD06,FernauR14}.
By replacing the condition ``for every $v \in S$'' with ``for every $v \in N(S)$,'' we obtain the definition of \emph{offensive alliances}.
A vertex set is a \emph{powerful alliance} if it is simultaneously a defensive alliance and an offensive alliance.
A defensive, offensive, or powerful alliance is \emph{global} if it is a dominating set.
All those concepts of alliances can be generalized to \emph{$r$-alliances} by adding a constant $r$
to the right-hand side of the inequality corresponding to \cref{eq:alliance} in their definitions.
Similarly to the case of the ordinary defensive alliance, we can express these variants and generalizations 
as \pname{$\GSOGLLin$ Model Checking}.
Therefore, the problem of finding these alliance of size at most $k$, named with the same rule as \pname{Defensive Alliance},
are fixed-parameter tractable parameterized by vertex integrity.

\begin{theorem}
(\pname{Global}) \pname{Defensive}/\pname{Offensive}/\pname{Powerful} $r$-\pname{Alliance}
are fixed-parameter tractable parameterized by vertex integrity.
\end{theorem}

\subsection{Known results (partially) captured by the meta-theorem}

\subsubsection{Bounded-degree deletion problems}

Recall that \pname{Bounded-Degree Vertex Deletion}~\cite{GanianKO21} is fixed-parameter tractable parameterized by vertex
integrity.\footnote{The parameter in~\cite{GanianKO21} is a generalization of vertex integrity.}
This problem asks for a minimum number $k$ of vertices to be removed to make the maximum degree at most $d$.
This problem can be easily expressed as \pname{$\MSOGLLin$ Model Checking}.
The formula $\varphi$ has a free vertex-set variable $X$ and a free edge-set variable $Y$
and the following constraints:
\begin{itemize}
  \item $|X| \ge |V| - k$, 
  \item the edge set of $G[X]$ is $Y$, 
  \item $G[X]$ has maximum degree at most $d$.
\end{itemize}

\subsubsection{Equitable partition problems}

Let $G = (V,E)$ be a graph with $n$ vertices, and let $r$ be a positive integer.
A partition of $V$ into $r$ sets $V_{1},\dots,V_{r}$ is an \emph{equitable $r$-partition} if
$|V_{i}| \in \{\lfloor n/r \rfloor, \lceil n/r \rceil\}$ for all $i \in [r]$.
Given $G$ and $r$, 
\pname{Equitable Coloring} asks whether 
$G$ admits an equitable $r$-partition $V_{1},\dots,V_{r}$ such that each $V_{i}$ is an independent set,
and \pname{Equitable Connected Partition} asks whether 
$G$ admits an equitable $r$-partition $V_{1},\dots,V_{r}$ such that each $G[V_{i}]$ is connected.

\pname{Equitable Coloring} is W[1]-hard parameterized by treedepth~\cite{FellowsFLRSST11}.
\pname{Equitable Connected Partition} is W[1]-hard parameterized
simultaneously by pathwidth\footnote{In \cref{sec:w1}, we strengthen the hardness result by replacing this part with treedepth.},
feedback vertex set number, and the number of parts $r$~\cite{EncisoFGKRS09}.

Both problems can be directly expressed as \pname{$\MSOGLLin$ Model Checking},
where the length of the formula depends on $r$ (see \cite{GanianO13}).
More generally, if the property asked for each $G[V_{i}]$ is expressible in $\MSO_{2}$, 
we can express the equitable partition problem as \pname{$\GSOGLLin$ Model Checking}.
This implies a weaker result that the problems are 
fixed-parameter tractable parameterized by both vertex integrity and $r$.

In~\cite{GimaHKKO21}, some problem specific approaches for dealing with unbounded $r$ were taken,
and \pname{Equitable Coloring} and \pname{Equitable Connected Partition} were shown to be 
fixed-parameter tractable parameterized solely by vertex integrity.
It would be interesting to find a unified way for handling 
the general equitable partition problem parameterized by vertex integrity only.

\subsubsection{Capacitated problems}
Let $G = (V,E, c)$ be a \emph{capacitated} graph with a 
\emph{capacity function} $c \colon V \to \mathbb{Z}$ such that $c(v) \le \deg(v)$ for each $v \in V$.
A set $C \subseteq V$ is a \emph{capacitated vertex cover} of $G$
if there is a mapping $f \colon E \to C$ such that $f(e)$ is an endpoint of $e$ for each $e \in E$
and $| \{e \in E \mid f(e) = v\} | \le c(v)$ for each $v \in C$.
That is, each vertex $v$ in a capacitated vertex cover can cover at most $c(v)$ incident edges.
Similarly, a set $D \subseteq V$ is a \emph{capacitated dominating set} of $G$
if there is a mapping $f \colon V \setminus D \to D$ such that $f(v) \in N(v) \cap D$ for each $v \in V \setminus D$
and $| \{v \in V \setminus D \mid f(v) = u\} | \le c(u)$ for each $u \in D$.
Namely, each vertex $v$ in a capacitated dominating set can dominate at most $c(v)$ neighbors.

The problems \pname{Capacitated Vertex Cover} and \pname{Capacitated Dominating Set}
ask whether a given capacitated graph has a capacitated vertex cover 
and a capacitated dominating set, respectively, of size at most $k$.
\pname{Capacitated Vertex Cover} is W[1]-hard parameterized by treedepth,
and \pname{Capacitated Dominating Set} is W[1]-hard parameterized by treedepth and $k$~\cite{DomLSV08}.
Both problems are fixed-parameter tractable parameterized by vertex integrity~\cite{GimaHKKO21}.

Expressing the capacitated problems as \pname{$\GSOGLLin$ Model Checking} is not very straightforward, but can be done as follows.
Let $H = (V \cup W, F)$ be the graph obtained from $G$ by subdividing each edge, where $W$ is the set of new vertices introduced.
Let $X$ be a subset of $V$ and $Y$ a subset of $F$.
We can express the following conditions as a $\GSOGLLin$ formula with free variables corresponding to $X$ and $Y$:
\begin{itemize}
  \item each $v \in X$ is incident to at most $c(v)$ edges in $Y$;
  \item no $v \in V \setminus X$ is incident to edges in $Y$;
  \item $|X| \le k$.
\end{itemize}
Intuitively, $X$ is (a candidate of) a solution of a capacitated problem
and $Y$ indicates how the capacity of each vertex in $X$ is assigned to neighboring edges or vertices.
We still need to express the conditions that in the original graph $G$,
$X$ satisfies the conditions for being a vertex cover or a dominating set.
For \pname{Capacitated Vertex Cover}, we add the following condition:
\begin{itemize}
  \item each $w \in W$ is incident to at least one edge in $Y$.
\end{itemize}
For \pname{Capacitated Dominating Set}, we add the following condition:
\begin{itemize}
  \item for each $v \in V \setminus X$, there exists a neighbor $w \in N_{H}(v)$ incident to some edge in $Y$.
\end{itemize}
Clearly, the conditions above correctly express the problems.

Observe that the last part of adding certain conditions for $X$ (and $Y$) would work for many other problems
as $\MSO_{2}$ properties on the graph $G$ can be expressed as $\MSO_{1}$ properties on its $1$-subdivision $H$
 (folklore, see also \cref{sec:ext}).
More precisely, the following problem can be expressed as \pname{$\GSOGLLin$ Model Checking}.

\begin{myproblem}
  \problemtitle{\pname{Capacitated $\MSO_{2}$ Model Checking}}
  \probleminput{A capacitated graph $G = (V,E,c)$, an $\MSO_{2}$ formula $\varphi(A,B)$, and an integer $k$.}
  \problemquestiontitle{Question}
  \problemquestion{Are there $X \subseteq V$ and $Y \subseteq E$ such that $|X| \le k$,
  $Y$ is a subset of the edges incident to $X$ with $|\{y \in Y \mid x \ \text{is an endpoint of} \ y\}| \le c(x)$ for each $x \in X$,
  and $(G, (X,Y)) \models \varphi$?}
\end{myproblem}

By \cref{thm:gsogl-alg}, we can conclude the following.
\begin{theorem}
\pname{Capacitated $\MSO_{2}$ Model Checking} is fixed-parameter tractable parameterized by 
$\vi(G)$ and $|\varphi|$.
\end{theorem}

\subsection{Known results (probably) not captured by the meta-theorem}

We have shown above that several known fixed-parameter tractability results can be obtained by applying our meta-theorem.
However, some of the known results seem not to be captured by the theorem.
We list them below with some points that make them difficult to be captured (which might be bypassed by some clever ideas).
\pname{Subgraph Isomorphism}~\cite{BodlaenderHKKOO20},
\pname{Maximum Common (Induced) Subgraph}~\cite{GimaHKKO21}, and
\pname{Locally Constrained Homomorphism}~\cite{BulteauDKOP22_arxiv} 
involve two graphs of unbounded size.
The definition of \pname{Imbalance}~\cite{GimaHKKO21} involves linear orderings of vertices.
\pname{Precoloring Extension}~\cite{GimaHKKO21} may use many, say $\Omega(n)$, colors in the input precoloring.
It would be interesting to further extend the study to capture (some of) these problems.

%%% W[1]-hardness of Equitable Connected Partition / treedepth
% -*- coding:utf-8 -*-

\section{W[1]-hardness of \textsc{Equitable Connected Partition} parameterized by treedepth}
\label{sec:w1}

As mentioned before, \pname{Equitable Connected Partition} is known to be W[1]-hard parameterized
simultaneously by pathwidth, feedback vertex set number, and the number of parts $r$~\cite{EncisoFGKRS09}.
In this section, we strengthen the W[1]-hardness by replacing pathwidth in the parameter with treedepth,
where the treedepth of a graph is always larger than or equal to its $\textrm{pathwidth} + 1$.
To the best of our knowledge, the complexity of \pname{Equitable Connected Partition} parameterized by treedepth was not known before.
(The reduction in~\cite{EncisoFGKRS09} uses long paths and thus the output instances have unbounded treedepth.)

The \emph{treedepth} $\td(G)$ of a graph $G = (V,E)$ is defined as follows:
\[
  \td(G) =
  \begin{cases}
    1 & |V| = 1,\\
    \max_{i \in [c]} \td(C_{i}) & G \textrm{ has $c\ge 2$ connected components } C_{1}, \dots, C_{c}, \\
    1 + \min_{v \in V} \td(G - v) & \textrm{otherwise}.
  \end{cases}
\]
Observe that $\td(G) \le \vi(G)$ for every graph $G$:
by removing a set $S \subseteq V(G)$, the treedepth decreases by at most $|S|$;
and $\td(G - S) \le \max_{C \in \mathrm{cc}(G-S)} |V(C)|$,
where $\mathrm{cc}(G-S)$ is the set of connected components of $G-S$.

\begin{theorem}
\textsc{Equitable Connected Partition} is W[1]-hard parameterized
simultaneously by treedepth, feedback vertex set number, and the number of parts.
\end{theorem}
\begin{proof}
We present a reduction from \textsc{Unary Bin Packing}.
Given a positive integer $t$ and $n$ positive integers $a_{1}, a_{2}, \dots, a_{n}$ in unary,
\textsc{Unary Bin Packing} asks whether the set $[n]$ can be partitioned into $t$ subsets $S_{1}, \dots, S_{t}$
such that $\sum_{i \in S_{j}} a_{i} = \frac{1}{t}\sum_{i \in [n]} a_{i}$ for each $j \in [t]$.
It is known that \textsc{Unary Bin Packing} is W[1]-hard parameterized by $t$~\cite{JansenKMS13}.

Let $\mathcal{I} = \langle t; a_{1}, a_{2}, \dots, a_{n} \rangle$ be an instance of \textsc{Unary Bin Packing}
with $\frac{1}{t} \sum_{i \in [n]} a_{i} = B$.
Observe that $B$ has to be an integer as otherwise $\mathcal{I}$ is a trivial no-instance.
From $\mathcal{I}$, we construct a graph as follows.
Take a complete bipartite graph with bipartition $(U,W)$
such that $U = \{u_{1}, \dots, u_{n}\}$ and $W = \{w_{1}, \dots, w_{t}\}$.
For each $u_{i}$, we attach $a_{i}-1$ pendants (that is, vertices of degree~$1$).
Also, for each $w_{i}$, we attach $2B-1$ pendants.
We call the obtained graph $G$. See \cref{fig:ecp}.
Note that $G$ has $3tB$ vertices.

\begin{figure}[htb]
\centering
\includegraphics[scale=1]{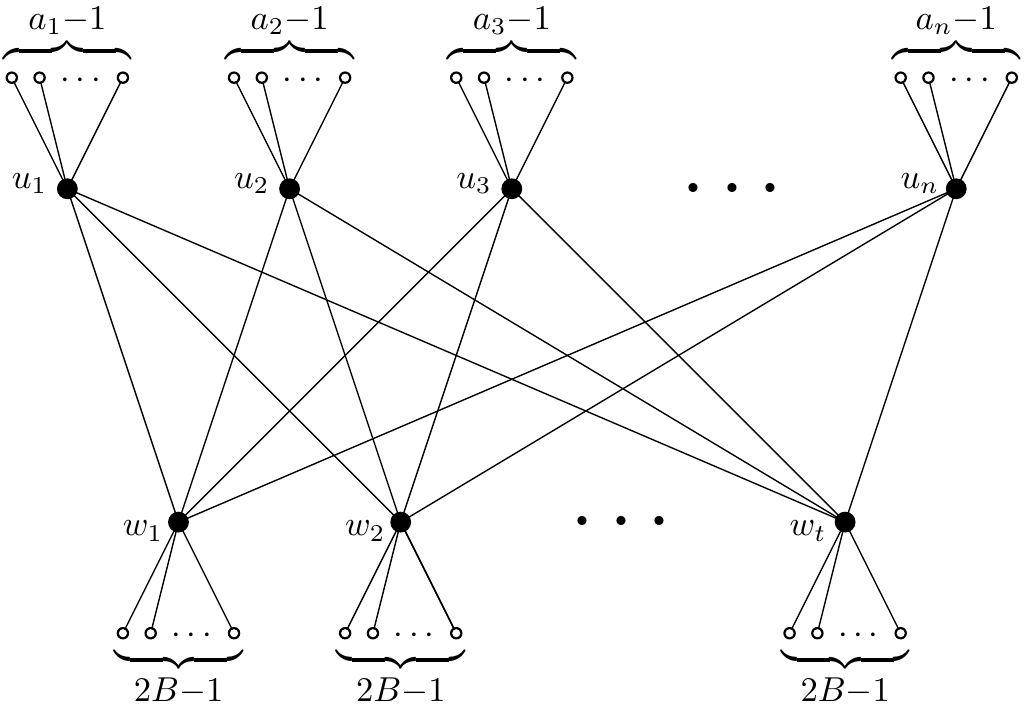}
\caption{The reduction from \textsc{Unary Bin Packing} to \textsc{Equitable Connected Partition}.}
\label{fig:ecp}
\end{figure}

We show that $\mathcal{I}$ is a yes-instance of \textsc{Unary Bin Packing}
if and only if $\langle G, t\rangle$ is a yes-instance of \textsc{Equitable Connected Partition}.
Observe that $G$ has treedepth at most $t + 2$
since after the removal of $W$ it becomes a disjoint union of stars and isolated vertices.
This also means that $W$ is a feedback vertex set.
Thus the equivalence of $\mathcal{I}$ and  $\langle G, t\rangle$ implies the theorem.

To show the only-if direction, assume that there is a partition $S_{1}, \dots, S_{t}$ of $[n]$
such that $\sum_{i \in S_{j}} a_{i} = B$ for each $j \in [t]$.
For each $j \in [t]$, let $V_{j}$ be the set formed by
$w_{j}$, the vertices in $\{u_{i} \mid i \in S_{j}\}$, and the pendants adjacent to them.
Clearly, $G[V_{j}]$ is connected and $|V_{j}| = 2B + \sum_{i \in S_{j}} a_{i} = 3B$.
Thus, the partition $V_{1}, \dots, V_{t}$ is a yes-certificate for $\langle G, t\rangle$.

To show the if direction, assume that there is a partition $V_{1}, \dots, V_{t}$ of $V(G)$
such that $G[V_{j}]$ is connected and $|V_{j}| = 3B$ for each $j \in [t]$.
Since $2B > 1$, a pendant and its unique neighbor belongs to the same set in the partition.
Hence, if we set $f(V_{j}) = \sum_{u_{i} \in V_{j}} a_{i} + \sum_{w_{h} \in V_{j}} 2B$,
then we have $f(V_{j}) = |V_{j}|$.
Since $f(V_{j}) = |V_{j}| = 3B$ for each $j \in [t]$, each $V_{j}$ includes exactly one vertex in $\{w_{1},\dots,w_{t}\}$.
This implies that $f(V_{j} \setminus \{w_{1},\dots,w_{t}\}) = B$,
and thus $\sum_{u_{i} \in V_{j}} a_{i} = B$.
By setting $S_{j} = \{i \mid u_{i} \in V_{j}\}$, we obtain a yes-certificate for $\mathcal{I}$.
\end{proof}

% \renewcommand{\refname}{References for the appendix}

% \putbib

% \end{bibunit}

\bibliography{msovi}

\end{document}